\documentclass[11pt]{article}

\usepackage{amssymb}

\usepackage{pstricks, pst-coil, pst-node, pst-tree, multido}

\newcommand {\sm} {\setminus}

\usepackage{graphicx}
\usepackage{latexsym} 
\usepackage{theorem} 
\newcommand{\qed}{\relax\ifmmode\hskip2em\Box\else\unskip\nobreak\hfill$\Box$\fi}

\newtheorem{defeng}{Definition}[section]
\newtheorem{theorem}[defeng]{Theorem}
\newtheorem{lemma}[defeng]{Lemma}

{\theoremstyle{break}\theorembodyfont{\rmfamily}}
{\theoremstyle{break}\theorembodyfont{\rmfamily}}

\newcommand{\tp}{\!-\!}

\newcounter{claim}

{\refstepcounter{claim}\vspace{1ex}\noindent{(\it\arabic{claim}){#1}{}}\it}{\vspace{1ex}}

    {\noindent {}{#1}{}}{ This proves~(\arabic{claim}).\vspace{1ex}}

 \newenvironment{proof}[1][]%
 {\noindent {\setcounter{claim}{0}\sc proof ---
    }{#1}{}}{\hfill$\Box$\vspace{2ex}}

\usepackage{ifpdf}

\ifpdf
\fi

\begin{document}

\title{The $k$-in-a-tree problem for graphs of girth at least~$k$}
\author{W. Liu\thanks{Universit\'e Grenoble 1 --- Joseph Fourier (France),
    email: wei@cmap.polytechnique.fr}  \ and  N. Trotignon\thanks{CNRS, LIAFA, Universit\'e Paris 7 ---
    Paris Diderot (France),\ email:
    nicolas.trotignon@liafa.jussieu.fr}\ }

\date{May 28, 2010}

\maketitle

\begin{abstract}
  For all integers $k\geq 3$, we give an $O(n^4)$ time algorithm for
  the problem whose instance is a graph $G$ of girth at least $k$
  together with $k$ vertices and whose question is ``Does $G$ contains
  an induced subgraph containing the $k$ vertices and isomorphic to a
  tree?''.

  This directly follows for $k=3$ from the three-in-a-tree algorithm
  of Chudnovsky and Seymour and for $k=4$ from a result of Derhy,
  Picouleau and Trotignon.  Here we solve the problem for $k\geq 5$.
  Our algorithm relies on a structural description of graphs of girth
  at least $k$ that do not contain an induced tree covering $k$ given
  vertices ($k\geq 5$).
\end{abstract}

\noindent AMS Mathematics Subject Classification: 05C75, 05C85, 05C05,
68R10, 90C35

\noindent Key words: tree, algorithm, three-in-a-tree, $k$-in-a-tree,
girth, induced subgraph.

\section{Introduction}

Many interesting classes of graphs are defined by forbidding induced
subgraphs, see~\cite{chudnovsky.seymour:excluding} for a survey.  This
is why the detection of several kinds of induced subgraphs is
interesting, see~\cite{leveque.lmt:detect} where many such problems
are surveyed.  In particular, the problem of deciding whether a graph
$G$ contains as an induced subgraph some graph obtained after possibly
subdividing prescribed edges of a prescribed graph $H$ has been
studied.  It turned out that this problem can be polynomial or
NP-complete according to $H$ and to the set of edges that can be
subdivided.  The most general tool for solving this kind of problems
(when they are polynomial) seems to be the \emph{three-in-a-tree}
algorithm of Chudnovsky and Seymour:

\begin{theorem}[see \cite{chudnovsky.seymour:theta}]
  \label{th:cs}
  Let $G$ be a graph and $x_1, x_2, x_3$ be three distinct vertices of
  $G$.  Deciding whether there exists an induced tree of $G$ that
  contains $x_1, x_2, x_3$ can be performed in time $O(n^4)$.
\end{theorem}

How to use three-in-a-tree is discussed
in~\cite{chudnovsky.seymour:theta} and further evidences of its
generality are given in~\cite{leveque.lmt:detect}.  The complexity of
four-in-a-tree is not known, and more generally of $k$-in-a-tree,
where $k\geq 4$ is a fixed integer.  But these problems are more
tractable when restrictions are given on the girth (length of a
smallest cycle) of the graph as suggested by Derhy, Picouleau and
Trotignon who proved:

\begin{theorem}[see \cite{nicolas.d.p:fourTree}]
  \label{th:dpt}
  Let $G$ be a triangle-free graph and $x_1, x_2, x_3, x_4$ be four
  distinct vertices of $G$. Deciding whether there exists an induced
  tree of $G$ that contains $x_1, x_2, x_3, x_4$ can be performed in time
  $O(nm)$.
\end{theorem}

Here, we study $k$-in-a-tree for graphs of girth at least $k$.  Note
that the problem is solved by the two theorems above for $k=3$ and
$k=4$.  For $k\geq 5$, we follow the method that has been already
succesful for Theorems~\ref{th:cs} and~\ref{th:dpt}: studying the
structure of a graph that does not contain the desired tree. It turns
out that in most of the cases, the structure is simple.  Note that the
proofs in the present work are independent
form~\cite{chudnovsky.seymour:theta,nicolas.d.p:fourTree}: we do not
use results from~\cite{chudnovsky.seymour:theta,nicolas.d.p:fourTree},
and as far as we can see, our results do not
simplify~\cite{chudnovsky.seymour:theta}
or~\cite{nicolas.d.p:fourTree}.

\begin{figure}[h]
\begin{center}
\includegraphics{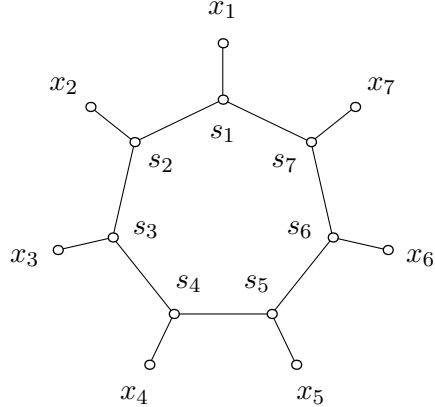}
\caption{a $k$-structure ($k=7)$\label{f:cycle}}
\end{center}
\end{figure}

We call \emph{$k$-structure} any graph obtained from the cycle on $k$
vertices by adding a pending path to each vertex of cycle, see
Section~\ref{sec:k} for a formal definition.  An example is shown in
Figure~\ref{f:cycle} which obviously does not contain an induced tree
covering the $k$ pending vertices.  The main result of
Section~\ref{sec:link} states that for $k\geq 3$, a graph of girth at
least $k$ that does not contain an induced tree covering $k$ given
vertices \emph{must contain} a $k$-structure.  The main result of
Section~\ref{sec:k} states that (with one exception, see below), if
the graph contains a $k$-structure, then the $k$-structure
\emph{decomposes} the graph, meaning that every vertex of the original
cycle is a cut-vertex of the graph.

\begin{figure}[h]
  \begin{center}
    \includegraphics{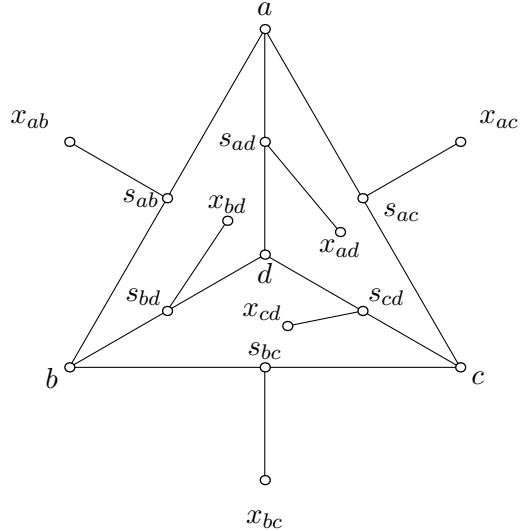}
    \caption{a $K_4$-structure \label{f:k4s}}
  \end{center}
\end{figure}

But there is a noteworthy exception that arises curiously only when
$k=6$.  The graph $G$ on Figure~\ref{f:k4s} is obtained from $K_4$ by
subdividing all edges once, and by adding a pending path to each
vertex of degree~2. This graph has girth 6.  Let $H$ be a connected
induced subgraph of $G$ that contains the 6 pending vertices.  We
claim that $H$ contains at least three vertices of degree 3 in $G$.
Otherwise, it does not contains at least 2 of them, so the pending vertex
whose neighbor is between these is isolated; a contradiction.  Hence, $H$
contains three vertices of degree 3 and a cycle of length 6 goes
through them.  Hence, no induced tree of $G$ can cover the 6 pending
vertices.  This is what we call a \emph{$K_4$-structure}.  The main
result of Section~\ref{sec:K4} states roughly that if a graph of
girth~6 contains a $K_4$-structure and if no induced tree covers the 6
pending vertices then the $K_4$-structure \emph{decomposes} the graph,
meaning that every pair of vertices of the original $K_4$ and every
vertex of degree~2 arising from the subdivisions is a cutset of the
graph.

Let us sum up the results.  Our main result, Theorem~\ref{th:main},
states that when $k\geq 5$, and $G$ is a connected graph of girth at
least $k$ together with $k$ vertices then either $G$ contains a
$k$-structure that decomposes $G$, or $k=6$ and $G$ contains a
$K_4$-structure that decomposes $G$, or $G$ contains an induced tree
covering the $k$ vertices.  All this leads to an $O(n^4)$-time
algorithm that decides whether a graph of girth at least $k$ contains
an induced tree that covers $k$ prescribed vertices.

\subsection*{Notation, convention, remarks}

We use standard notation from~\cite{gibbons:agt}.  Since we use only
induced subgraphs, we say that $G$ \emph{contains} $H$ when $H$ is an
induced subgraph of $G$.  Also, by \emph{tree of $G$} we mean induced
subgraph of $G$ that is a tree.  By \emph{path} we mean induced path.
In complexity of algorithms, $n$ stands for the number of vertices of
the input graph and $m$ for the number of its edges.  We call
\emph{terminal} of a graph any vertex of degree one.  Solving
$k$-vertices-in-a-tree or $k$-terminals-in-a-tree are equivalent
problems, because if $k$ vertices $x_1, \dots, x_k$ of graph $G$ are
given, we build the graph $G'$ obtained from $G$ by adding a pending
neighbor $y_i$ to $x_i$, $i=1, \dots, k$.  An induced tree of $G$
covers $x_1, \dots, x_k$ if and only an induced tree of $G'$ covers
$y_1, \dots, y_k$.  Hence, in the rest of the paper we assume for
convenience that the vertices to be covered are all terminals.

\section{Linking a vertex to a tree}
\label{sec:link}

Recall that a \emph{terminal} in a graph is a vertex of degree~1.  A
\emph{branch-vertex} is a vertex of degree at least~3.  The following
is a basic fact whose proof is omited.

\begin{lemma}
  \label{l:descTree}
  A tree $T$ with $k$ terminals contains at most $k-2$
  branch-vertices.  Moreover if $T$ contains exactly $k-2$
  branch-vertices then every branch-vertex is of degree 3.
\end{lemma}

\begin{lemma}
  \label{l:linkTree}
  Let $k, l$ be integers such that $k\geq 3$ and $2 \leq l \leq k$.
  Let $G$ be a graph of girth at least $k$ and $x_1, \dots, x_l$ be
  $l$ distinct terminals of $G$.  Let $T$ be an
  induced tree of $G$ whose terminals are $x_1, \dots, x_{l-1}$.  Let
  $Q$ be a path from $x_l$ to $w$ such that $w$ has at least one
  neighbor in $T$ and no vertex of $Q\sm w$ has neighbors in $T$. Then
  one and only one of the following outcomes holds:

  \begin{itemize}
  \item $T \cup Q$ contains a tree of $G$ that covers $x_1, \dots,
    x_l$.
  \item $k=l$. Moreover, $T$ and $Q$ can be described as follows (up
    to a relabelling of $x_1, \dots, x_{k-1}$):
    \begin{enumerate}
    \item $T$ is the union of $k-1$ vertex-disjoint paths $s_1 \tp
      \cdots \tp x_1$, $s_2 \tp \cdots \tp x_2$, \dots, $s_{k-1} \tp
      \cdots \tp x_{k-1}$;
    \item the only edges between these paths are such that $s_1\tp s_2
      \tp \cdots \tp s_{k-1}$ is a path;
    \item $N_T(w) = \{s_1, s_{k-1}\}$.
    \end{enumerate}
  \end{itemize}

  This is algorithmic in the sense that when $T$ and $Q$ are given,
  the tree of the first outcome or the relabelling of the second can
  be computed in time $O(n^3)$.
\end{lemma}

\begin{proof}
  Clearly, at most one of the outcomes holds (because if the second
  holds then no tree of $T\cup Q$ can cover $x_1, \dots, x_l$).  Let
  us prove that at least one of the outcomes holds.
  
  Let $W = \{w_1, \dots, w_i\}$ be the set of the neighbors of $w$ in
  $T$.  If $i=1$ then $T\cup Q$ is a tree that covers $x_1, \dots,
  x_{l}$ so let us suppose that $i\geq 2$.  Let us call a \emph{basic
    path} any subpath of $T$ linking two distinct vertices of $W$ and
  with no interior vertices in $W$.  All the basic paths are on at
  least $k-1$ vertices because the girth of $G$ is at least $k$.  Now
  we consider two cases:

  \noindent{\bf Case 1:} for all basic paths $R$ of $T$ there exists
  an interior vertex $v_R$ of $R$ that has degree two in $T$.  Then,
  let $S \leftarrow T\cup Q$.  For all basic paths $R$, if $R\subseteq
  S$, then let $v_R$ be a vertex of degree two (in $T$) of $R$, let $S
  \leftarrow S\sm \{v_R\}$ and go the next path $R$.  At the end of
  this loop, one vertex of degree two is deleted from all basic paths.
  Remark that one vertex $v_R$ can be contained in several basic
  paths.  Hence, $S$ contains no more cycle, but is still connected
  because the deleted vertices have all degree~2 and exactly one is
  deleted in each basic path.  Hence, we obtain a tree $S$ that covers
  $x_1, \dots, x_{l}$.  This takes time $O(n^3)$ because we enumerate
  all the pairs $w_i, w_j$ to find the basic paths.

  \noindent{\bf Case 2:} we are not in Case~1, so there exists a basic
  path $R$ whose interior vertices are all of degree at least~3 in
  $T$.  Then, since $T$ has $l-1$ terminals, Lemma~\ref{l:descTree}
  says that it has at most $l-3$ branch-vertices.  On the other hand,
  since a basic path is on at least $k-1$ vertices (because the girth
  is at least~$k$), $R$ contains at least $k-3$ branch-vertices of
  $T$.  So in fact, because $l\leq k$, we have $k=l$ and $R$ contains
  all the $k-3$ branch-vertices of $T$.  Since $R$ has no interior
  vertex of degree~2, in fact $R$ contains $k-1$ vertices.  We name
  $s_1,\cdots ,s_{k-1}$ the vertices of $R$.  Note that $w$ is
  adjacent to $s_1$ and $s_{k-1}$ because $R$ is a basic path.  In
  particular, $s_1$ and $s_{k-1}$ are not terminals of~$G$.

  For all $1\leq i \leq k-1$, $s_i$ is a cutvertex of $T$ that
  isolates one terminal among $x_1, \dots, x_{k-1}$ from all the other
  terminals.  Up to relabelling, we suppose that this terminal is
  $x_i$.  We name $P_i$ the unique path of $T$ between $x_i$ and
  $s_i$.

  Note that $w$ is not adjacent to $s_2, \dots, s_{k-2}$ (because $R$
  is a basic path).  So the second outcome of our lemma holds, unless
  $w$ has at least one neighbor in some $P_i \sm s_i$.  For $i=1,
  \dots, k-1$ , we let $s'_i$ be the neighbor of $s_i$ along $P_i$, if
  $w$ has a neighbor in $P_i$ then we name $w_i$ the neighbor of $w$
  closest to $x_i$ along $P_i$ and if no such neighbor exists, we put
  $w_i = s_i$.

  Suppose that for all $i=1, \dots, k-1$ we have $w_i\neq s'_i$.
  Then, the paths $x_i \tp P_i \tp w_i$, $i=1, \dots, k-1$ together
  with $Q$ and $s_1, \dots, s_{k-1}$ form a graph with a unique cycle:
  $w s_1 \dots s_{k-1}w$.  By deleting a vertex $s_j$ such that
  $w_j\neq s_j$, we obtain a tree that covers $x_1, \dots, x_k$.

  Hence, we may assume that for some $i$, $w_i=s'_i$ and up to symmetry
  we suppose $i\leq k/2$.  Then $ws_1\dots s_i s'_iw$ is a cycle on
  $i+2$ vertices, so $i+2\geq k$ because of the girth.  Hence, $k-2\leq
  k/2$, so $k\leq 4$.  Then the paths $x_j \tp P_j \tp w_j$, $j=1,
  \dots, k-1$, together with $Q$ form a tree that covers $x_1, \dots,
  x_k$.
\end{proof}

A graph is a \emph{$k$-structure} with respect to $k$ distinct
terminals $x_1, \dots, x_k$ if it is made of $k$ vertex-disjoints paths
of length at least one $P_1 = x_1 \tp \cdots \tp s_1$, \dots, $P_k =
x_k\tp \cdots \tp s_k$ such that the only edges between them are
$s_1s_2$, $s_2s_3$, \dots, $s_{k-1}s_k$, $s_ks_1$.

\begin{lemma}
  \label{l:firststep}
  Let $k\geq 3$ be an integer.  Let $G$ be a connected graph of girth
  at least $k$ and $x_1, \dots, x_l$ be $l$ terminals where $1\leq l \leq
  k$.  Then either $G$ contains a tree that covers the $l$ terminals
  or $l=k$ and $G$ contains a $k$-structure with respect to $x_1,
  \dots, x_k$.

  This is algorithmic in the sense that we provide an $O(n^4)$
  algorithm that finds the tree or the $k$-structure. 
\end{lemma}

\begin{proof}
  We suppose that $k$ is fixed and we prove the statement by induction
  on $l$.  For $l=1$ and $l=2$, the lemma is clear: a tree exists (for
  instance, a shortest path linking the two terminals).  Suppose the
  lemma holds for some $l-1<k$ and let us prove it for~$l$.  By the
  induction hypothesis there exists an induced tree $T$ of $G$ that
  covers $x_1, \dots, x_{l-1}$.  Let $Q$ be a path from $x_l$ to some
  vertex $w$ that has neighbors in $T$, and suppose that $Q$ is
  minimal with respect to this property.  Then, no vertex of $Q\sm w$
  has a neighbor in $T$.

  We apply Lemma~\ref{l:linkTree}.  If the first outcome holds, we
  have our tree.  Otherwise, $T\cup Q$ is a $k$-structure.  All this can be
  implemented in time $O(n^4)$ because terminals are taken one by one,
  there are at most $n$ of them and for each of them we rely on basic
  subroutines like BFS (Breadth First Search, see~\cite{gibbons:agt})
  to find $Q$ and on the $O(n^3)$ algorithm of Lemma~\ref{l:linkTree}.
\end{proof}

\section{The $K_4$-structure}
\label{sec:K4}

A graph is a \emph{$K_4$-stucture} with respect to 6 distinct
terminals $x_{ab}, x_{ac}, x_{ad}, x_{bc}, x_{bd}, x_{cd}$ if it is
made of 6 vertex-disjoints paths of length at least one $P_{ab} =
x_{ab}\tp \cdots \tp s_{ab}$, $P_{ac} = x_{ac}\tp \cdots \tp s_{ac}$,
$P_{ad} = x_{ad}\tp \cdots \tp s_{ad}$, $P_{bc} = x_{bc}\tp \cdots \tp
s_{bc}$, $P_{bd} = x_{bd}\tp \cdots \tp s_{bd}$, $P_{cd} = x_{cd}\tp
\cdots \tp s_{cd}$ and four vertices $a, b, c, d$ such that the only
edges between them are $as_{ab}$, $as_{ac}$, $as_{ad}$, $bs_{ab}$,
$bs_{bc}$, $bs_{bd}$, $cs_{ac}$, $cs_{bc}$, $cs_{cd}$, $ds_{ad}$,
$ds_{bd}$, $ds_{cd}$.  (See Figure~\ref{f:k4s}.)  We put $X= \{x_{ab},
x_{ac}, x_{ad}, x_{bc}, x_{bd}, x_{cd}\}$.

We use the following ordering of the vertices $a$, $b$, $c$, $d$: $a <
b < c < d$.  We say that a $K_4$-structure $K$ in a graph $G$
\emph{decomposes} $G$ if the two following conditions hold:
\begin{enumerate}
\item\label{i:d1} for all $i, j$ such that $a\leq i < j \leq d$, $\{i, j\}$ is a
 cutset of $G$ that separates $x_{ij}$ from $X\sm \{x_{ij}\}$;
\item\label{i:d2} for all $i, j$ such that $a\leq i< j \leq d$,
  $\{s_{ij}\}$ is a cutset of $G$ that separates $x_{ij}$ from $X\sm
  \{x_{ij}\}$.
\end{enumerate}

\begin{lemma}
  \label{l:k4}
  If a graph $G$ of girth 6 contains a $K_4$-structure $K$ with
  respect to 6 terminals $x_{ab}, x_{ac}, x_{ad}, x_{bc}, x_{bd},
  x_{cd}$ then one and only one of the following outcomes holds:
  \begin{itemize}
  \item $K$ decomposes $G$;
  \item $G$ contains a tree that covers $x_{ab}, x_{ac}, x_{ad},
    x_{bc}, x_{bd}, x_{cd}$.
  \end{itemize}
  This is algorithmic in the sense that if $K$ is given, testing
  whether $K$ decomposes $G$ or outputing the tree can be performed in
  time $O(n^4)$.
\end{lemma}

\begin{proof}
  Let us first check that at most one of the output holds. Suppose
  that the first outcome holds, and let $H$ be a connected induced
  subgraph of $G$ covering $X$. Then $H$ must contain at least three
  vertices among $a, b, c, d$, because if it fails to contain two of
  them, say $a, b$, then $x_{ab}$ is isolated from the rest of the
  graph because of Condition~\ref{i:d1}.  Hence, we may assume that $H$
  contains $a, b, c$.  Also, because of Condition~\ref{i:d2}, $H$ must
  contain $s_{ab}$, $s_{bc}$ and $s_{ac}$.  Hence, $H$ contains the cycle
  $as_{ab}bs_{bc}cs_{ac}a$.  Hence, $H$ cannot be a tree, so the
  second outcome fails.

  Now let $H$ be an induced subgraph of $G$ that contains $K$ and such
  that $K$ decomposes $H$ ($H$ exists since $K$ decomposes $K$).  We
  show that for any vertex $v$ of $G\sm H$, $H\cup \{v\}$ either is
  decomposed by $K$ or contains a tree covering $X$.  This will prove
  the theorem by induction and will be the description of an $O(n^4)$
  algorithm since for each $v$, the proof gives the way to actually
  build the tree when there is one by calling the algorithm of
  Lemma~\ref{l:linkTree} and searching the graph (with BFS for
  instance).  Note also that testing whether $K$ decomposes some graph
  can be performed in linear time by 12 checks of connectivity.

  Suppose that $H\cup \{v\}$ is not decomposed by $K$.  From the
  definition of decomposition, there are two cases:

  \noindent{\bf Case 1:} Condition~\ref{i:d1} fails.  Up to symmetry,
  we suppose that $\{a, b\}$ is a not cutset of $H\cup \{v\}$ that
  separates $x_{ab}$ from $X\sm \{x_{ab}\}$.  Let $Y$ (resp. $Z$) be
  the connected component of $H\sm \{a, b\}$ that contains $x_{ab}$
  (resp. that contains $K' = K\sm (P_{ab}\cup \{a, b\})$).  Hence, $v$
  has a neighbor in $Y$ and a neighbor in $Z$.  Let $Q$ be a shortest
  path in $Y\cup Z \cup \{v\}$ from $x_{ab}$ to some vertex $w$ that
  has a neighbor in $K'$.  Note that $Q$ must go through $v$.  Because
  $K'$ is a tree that covers $X\sm \{x_{ab}\}$, we may apply
  Lemma~\ref{l:linkTree} to $K'$ and $Q$ in $Q \cup K'$. Hence, either we
  find the tree or $w$ has exactly two neighbors in $K'$ that have
  degree 2 in $K'$ and that are adjacent to $c$ or $d$.  Since the
  girth is 6, we may assume up to symmetry that these two neighbors
  are $s_{bc}$ and $s_{ad}$.  Because of the girth 6, $w$ is not
  adjacent to $a$, $b$ and $s_{ab}$.

  If $w$ has a neighbor in $P_{ab}$, we let $P$ be a shortest path
  from $w$ to $x_{ab}$ in $P_{ab} \cup \{w\}$.  Otherwise, we let $P =
  P_{ab}$.  We observe that $P \cup \{a, d, w\} \cup P_{ac} \cup
  P_{ad} \cup P_{bc} \cup P_{bd} \cup P_{cd}$ is a tree that covers
  $X$.

  \noindent{\bf Case 2:} Condition~\ref{i:d1} is satisfied but
  Condition~\ref{i:d2} fails.  Up to symmetry, we suppose that
  $\{s_{ab}\}$ is a not cutset of $H\cup \{v\}$ that separates
  $x_{ab}$ from $X\sm \{x_{ab}\}$.  Let us consider a path $R$ in
  $H\cup \{v\}$ from $x_{ab}$ to some vertex in $K\sm\{P_{ab}\}$ and
  let us suppose $R$ is minimal with respect to this property. Since
  Condition~\ref{i:d1} is satisfied, $R$ must be from $x_{ab}$ to $a$
  or $b$ ($a$ say).  Note that the neighbor of $a$ along $R$ cannot be
  adjacent to $b$ (or there is a cycle on 4 vertices).  We observe
  that $R\cup (K \sm (\{d\} \cup P_{ab}))$ is a tree that covers $X$.
\end{proof}

\section{The $k$-structure}
\label{sec:k}

For $k$-structures, we assume that notation like in the definition is
used.  We put $X= \{x_1, \dots, x_k\}$.  We say that a $k$-structure
$K$ in a graph $G$ \emph{decomposes} $G$ if for all $i$ such that
$1\leq i \leq k$, $\{s_i\}$ is a cutset of $G$ that separates $x_{i}$
from $X\sm \{x_{i}\}$.

\begin{lemma}
  \label{l:k}
  Let $k\geq 5$ be an integer. If a graph $G$ of girth at least $k$
  contains a $k$-structure $K$ with respect to $k$ terminals $x_1,
  \dots, x_k$ then one of the following outcomes holds:
  \begin{itemize}
  \item $K$ decomposes $G$;
  \item $k=6$ and there exists a vertex $v$ of $G\sm K$ such that
    $K\cup\{v\}$ is a $K_4$-structure with respect to $x_1, \dots,
    x_6$;
  \item $G$ contains a tree that covers $X$.
  \end{itemize}
  This is algorithmic in the sense that testing whether $K$ decomposes
  $G$ or outputing the tree or outputing a ralebelling showing that
  $K\cup\{v\}$ is a $K_4$-structure can be performed in time $O(n^4)$.
 \end{lemma}

\begin{proof}
  Let $H$ be an induced subgraph of $G$ that contains $K$ and such
  that $K$ decomposes $H$ ($H$ exists since $K$ decomposes $K$). We
  show that for any vertex $v$ of $G\sm H$, $H\cup \{v\}$ either
  satisfies the first outcome or is a $K_4$-structure or contains a
  tree covering $X$.  This will prove the theorem by induction and be
  the description of an $O(n^4)$ algorithm since for each~$v$, the
  proof gives the way to actually build the tree or the relabelling by
  calling the algorithm of Lemma~\ref{l:linkTree} and searching the
  graph (with BFS for instance).  Note also that testing whether $K$
  decomposes some graph can be performed in time $O(km)$, or $O(nm)$
  since $k\leq n$, by $k$ checks of connectivity.

  Suppose that $H\cup \{v\}$ is not decomposed by $K$.  Let $Y$
  (resp. $Z$) be the connected component of $H\sm \{s_1\}$ that
  contains $x_{1}$ (resp. that contains $K' = K\sm P_{1}$).  Up to
  symmetry, we may assume that $v$ has a neighbor in $Y$ and a
  neighbor in $Z$.  Let $Q$ be a shortest path in $Y\cup Z \cup \{v\}$
  from $x_1$ to some vertex $w$ that has a neighbor in $K'$.  Note
  that $Q$ must go through~$v$.  Because $K'$ is a tree that covers
  $X\sm \{x_{1}\}$, we may apply Lemma~\ref{l:linkTree} to $K'$ and
  $Q$ in $Q \cup K'$.  Hence, either we find the tree or $w$ has exactly
  two neighbors in $K'$ and $N_{K'}(w)$ must be one of the folowing:
  $\{s_{2}, s_{k}\}$, $\{s_{2}, s'_{k-1}\}$, $\{s'_{3}, s_{k}\}$,
  $\{s'_{3}, s'_{k-1}\}$ where $s'_i$ denotes the neighbor of $s_i$
  along $P_i$.

  When $N_{K'}(w) = \{s_{2}, s_{k}\}$, we observe that $s_2s_1s_kw$ is
  a \emph{square}, i.e. a cycle on 4 vertices, contradicting our
  assumption on the girth.

  When $N_{K'}(w) = \{s_{2}, s'_{k-1}\}$ (or symmetrically $\{s'_{3},
  s_{k}\}$), then $w$ is not adjacent to $s'_1$ (otherwise $s'_1s_1s_2w$ is
  a square).  If $w$ has a neighbor in $P_{1}$, we let $P$ be a
  shortest path from $w$ to $x_{1}$ in $P_{1} \cup \{w\}$.  Otherwise, we
  let $P = P_{1}$.  We observe that $\{w\} \cup P \cup (K' \sm
  \{s_{k-1}\})$ is a tree that covers $X$.

  We are left with the case when $N_{K'}(w) = \{s'_{3}, s'_{k-1}\}$.
  Suppose first that $w$ has no neigbhor in~$P_1$.  Then $\{w\} \cup K
  \sm \{s_3\}$ is a tree that covers $X$.  Suppose now that $w$ has a
  neighbor in $P_1\sm \{s_1, s'_1\}$.  We let $P$ be a shortest path
  from $w$ to $x_{1}$ in $\{w\} \cup (P_{1} \sm \{s_1, s'_1\})$.  If
  $ws_1 \notin E(G)$ then $P \cup \{s_1\} \cup (K\sm (P_1 \cup
  \{s_3\}))$ induces a tree that covers~$X$.  If $ws_1 \in E(G)$ then
  we observe that $P \cup \{s_1\} \cup (K\sm (P_1 \cup \{s_3,
  s_{k-1}\}))$ induces a tree that covers~$X$.  

  So we may assume that $N_{P_1}(w)$ is one of $\{s_1\}$,
  $\{s'_1\}$. If $N_{P_1}(w) = \{s_1\}$ then $s_1ws'_3s_3s_2$ is a
  $C_5$ so $k=5$ because of the girth assumption.  Hence $\{w\} \cup K
  \sm \{s_3, s_4\}$ is a tree that covers $X$.  So we are left with
  the case when $N_{P_1}(w) = \{s'_1\}$.  Then $ws'_1s_1s_2s_3s'_3$ is
  a $C_6$, so $k=5$ or $6$ because of the girth.  If $k=5$ then $\{w\}
  \cup K \sm \{s_3, s_4\}$ is a tree that covers $X$.  If $k=6$ then
  $K\cup \{w\}$ is a $K_4$-structure as shown by the following
  relabelling: $x_{ab} \leftarrow x_1$, $x_{ac} \leftarrow x_3$,
  $x_{ad} \leftarrow x_5$, $x_{bc} \leftarrow x_2$, $x_{bd} \leftarrow
  x_6$, $x_{cd} \leftarrow x_4$, $a \leftarrow w$, $b \leftarrow s_1$,
  $c \leftarrow s_3$, $d \leftarrow s_5$, $s_{ab} \leftarrow s'_1$,
  $s_{ac} \leftarrow s'_3$, $s_{ad} \leftarrow s'_5$, $s_{bc}
  \leftarrow s_2$, $s_{bd} \leftarrow s_6$, $s_{cd} \leftarrow s_4$.
\end{proof}

\section{The main result}
\label{sec:main}

\begin{theorem}
\label{th:main}
Let $k\geq 5$ be an integer.  Let $G$ be a connected graph of girth at
least $k$ and $x_1, \dots, x_k$ be terminals of $G$.  Then one and
only one of the following holds:

\begin{itemize}
\item $G$ contains $k$-structure $K$ with respect to $x_1, \dots, x_k$
  and $K$ decomposes~$G$; 
\item $k=6$, $G$ contains a $K_4$-structure $K$ with respect to $x_1, \dots,
  x_6$ and $K$ decomposes $G$;
\item $G$ contains a tree covering $x_1, \dots, x_k$.
\end{itemize}

This is algorithmic in the sense that we provide an algorithm that
output the tree or the structure certifying that no such tree exists
in time $O(n^4)$.
\end{theorem}

\begin{proof}
  By Lemma~\ref{l:firststep}, we can output a tree covering $X$ or a
  $k$-structure of $G$ in time $O(n^4)$.  If a $k$-structure $K$ is
  ouptut, then by Lemma~\ref{l:k}, we can check whether $K$ decomposes
  $G$ (in which case no tree exists) or find a tree, or find a
  $K_4$-structure $K'$.  In this last case, by Lemma~\ref{l:k4}, we can
  check whether $K'$ decomposes $G$ or find a tree. 
\end{proof}

\begin{theorem}
\label{th:algo}
Let $k\geq 3$ be an integer.  Let $G$ be a connected graph of girth at
least $k$ and $x_1, \dots, x_k$ be vertices of $G$.  Deciding whether
$G$ contains an induced tree covering $x_1, \dots, x_k$ can be
performed in time $O(n^4)$. 
\end{theorem}

\begin{proof}
  Follows from Theorem~\ref{th:cs} for $k=3$, from
  Theorem~\ref{th:dpt} for $k=4$ and from Theorem~\ref{th:main} for
  $k\geq 5$.
\end{proof}

\noindent{\bf Remark:} In all the proofs above for $k\geq 5$, we use
very often that the input graph contains no triangle and no square.
Forbidding longer cycles is used less often.  This suggests that the
$k$-in-a-tree problem might be polynomial for graphs with no triangle
and no square.  We leave this as an open question.

\end{document}